\documentclass[12pt,reqno]{amsart}
\usepackage{amsfonts, amsmath, amssymb, amsthm,bm,bbm}
\usepackage{bm}
\usepackage{color}
\usepackage{stmaryrd}
\usepackage{graphicx}
\usepackage{tcolorbox}
\tcbuselibrary{breakable}
\usepackage[pagebackref, colorlinks = true, linkcolor = blue, urlcolor  = blue, citecolor = red]{hyperref}
\usepackage[margin=1in, marginparwidth=1.5cm]{geometry}
\usepackage{appendix}

\DeclareMathOperator{\Tr}{Tr}

\newcommand{\nocontentsline}[3]{}
\newcommand{\tocless}[2]{\bgroup\let\addcontentsline=\nocontentsline#1{#2}\egroup}

\newcommand{\be}{\begin{equation}}
\newcommand{\ee}{\end{equation}}
\newcommand{\bea}{\begin{eqnarray}}
\newcommand{\eea}{\end{eqnarray}}

\newcommand{\cM}{\mathcal{M}}

\newcommand{\cP}{\mathcal{P}}

\newcommand{\bM}{\mathbb{M}}

\newcommand{\bE}{\mathbb{E}}
\newcommand{\bC}{\mathbb{C}}
\newcommand{\bN}{\mathbb{N}}

\newcommand{\cH}{\mathcal{H}}

\newcommand{\un}{\mathbbm{1}}

\newtheorem{theorem}{Theorem}[section]

\newtheorem{lemma}[theorem]{Lemma}
\newtheorem{definition}[theorem]{Definition}

\newtheorem*{theorem*}{Theorem}
\newtheorem*{theo31*}{Theorem \ref{thm:moments-fixed-N} in the maps formulation}

\newtheorem{proposition}[theorem]{Proposition}
\newtheorem{remark}[theorem]{Remark}

\begin{document}

\title[Asymptotic resolvents of a product of marginals of a random tensor]{Asymptotic resolvents of a product of two marginals of a random tensor}

\author{{Stephane~Dartois}}\email{stephane.dartois@unimelb.edu.au}
\address{School of Mathematics and Statistics, 
University of Melbourne, Victoria 3010, Australia}
\maketitle
\begin{abstract}
    Random tensors can be used to produce random matrices. This idea is, for instance, very natural when one studies random quantum states with the aim of exploring properties that are generically true, or true with some probability. We hereby study the moments generating function, in the sense of the Stieltjes transform - \textit{i.e.} the resolvent -, of a random matrix defined as a product of two different marginals of the same random tensor. We study the resolvent in two different asymptotical regimes.
    In the first regime, the resolvent is easily computed thanks to freeness results of the two different marginals obtained in \cite{DLN} and straightforward application of free harmonic analysis. In the second regime, we show that the resolvent satisfies an algebraic equation of degree six. This algebraic equation possesses roots whose expressions can be given explicitly in terms of radicals. We obtain this result by using an enumerative combinatorics approach. One of the interesting aspects of the second regime is that the corresponding probability density function interpolates between the square of a Mar\u cenko-Pastur and the free multiplicative square of a Mar\u cenko-Pastur law.
\end{abstract}
\vspace{3cm}
\noindent{\bf Keywords:} random quantum states, random tensors, random matrix products, free probability, resolvent, combinatorial maps, enumeration, marginals
\section{Introduction}
We hereby address a question triggered by \cite{DLN}. In this work the authors consider the properties of the moments of the products of different marginals of a random tensor in a variety of asymptotical regimes. Their study is motivated by quantum information questions as well as by the beautiful underlying combinatorics. One of the question they introduce is the following. Consider a random tensor $X\in \mathcal{H}_A\otimes \mathcal{H}_B\otimes \mathcal{H}_C\otimes \mathcal{H}_D$ with Gaussian entries, where $\mathcal{H}_A, \mathcal{H}_B, \mathcal{H}_C, \mathcal{H}_D$ are seen as Hilbert spaces of quantum systems. What is the behavior of the different marginals of $X$ seen as random matrices? In their paper the authors state different freeness results in different regimes, namely when the dimensions of all the Hilbert spaces $\mathcal{H}_B, \mathcal{H}_C$ are sent to infinity or when the dimensions of $\mathcal{H}_A, \mathcal{H}_D$ go to infinity and the dimensions of $\mathcal{H}_B, \mathcal{H}_C$ stay finite fixed to a value $m$. \\

In this letter we propose a method to compute the asymptotics of the resolvent of a random matrix $P$ made of a product of the two marginals of $X$, $$
V_{AB}=[\un \otimes \un \otimes\Tr\otimes\Tr](X\otimes X^{\vee}), \ V_{AC}=[\un \otimes \Tr \otimes\un\otimes\Tr](X\otimes X^\vee),$$ where $\textrm{dim }\mathcal{H}_B = \textrm{dim } \mathcal{H}_C =m$. Our main results, Theorem \ref{thm:main-thm}, states that the resolvent satisfies an (explicit) algebraic equation in the asymptotical regime where both $\textrm{dim }\mathcal{H}_A = N_A, \ \textrm{dim } \mathcal{H}_D =N_D$ are sent to infinity in a correlated way while $\textrm{dim }\mathcal{H}_B = \textrm{dim } \mathcal{H}_C =m$ is set to be fixed, finite. Interestingly, as was already noticed in the work \cite{DLN}, the parameter $m$ can be used to interpolate between an eigenvalue density that is the square of Mar\u cenko-Pastur ($m=1$) and the free multiplicative square of Mar\u cenko-Pastur ($m=\infty$). However, in \cite{DLN}, no analytics results were given for general values of $m$. The general values of $m$ were only explored numerically. In particular, it is obvious from our work that one can extend the result to non integer values of $m$.\\

This letter adds new interesting results to the literature on products of random matrices. That is of importance as products of random matrices are ubiquitous in a number of contemporary fields \cite{Jesper-thesis}, for example communication engineering \cite{TV04}, the analysis of algorithms \cite{Tr15}, and deep learning \cite{PW17}. The study of products of random matrices possesses many interesting combinatorial aspects. This is already transparent in the work \cite{DLN} where, for instance, discrete ramified coverings and projections appear naturally as a way to express the moments of marginals of random tensors. However it is a general feature which appears for instance in the (far from exhaustive) list of works \cite{DR-NarayanaWishart, lenczewski2013multivariate,Dubach-Peled}. In these works one can find relations to gluings of polygons and discretizations of Riemann surfaces with various local constraints, generalizations of Catalan and Fuss-Catalan numbers and of Narayana statistics.
The random matrices considered here are rather exotic as they come as marginals of the same random tensor. We here tell a combinatorial story, and we use combinatorial decomposition techniques to obtain the equation satisfied by the asymptotic of the resolvent. This approach seems to be best suited in this case due to the tensorial origin of the random matrices. Moreover, the combinatorics of this second asymptotical regime is both interesting (leading to a non-trivial spectral curve) and simple enough to be unraveled completely. \\
This problem also fits in a larger collection of works on random tensors and their applications. Indeed, random tensors and tensor models have been used to devise quantum gravity models \cite{delporte2018tensor, gurau2012complete}. A large amount of random matrices applications to quantum information can be seen as applications of random tensors to quantum information, \cite{DLN, collins2016random}. Finally, one for instance finds applications to turbulence \cite{dartois2018melonic} and machine learning/data analysis, \cite{arous2017landscape, arous2018algorithmic}. This makes new technical results on random tensors related problems potentially far reaching.\\

In order to supplement this work, interesting future prospects could involve: 1. Obtaining results on fluctuations at macroscopic scales, that is linear statistics computations. 2. Precise quantitative results for general marginals of random tensors in $n$-partite Hilbert spaces. 3. Generalization of the current product random matrices problem to a matrix model problem, that, if the corresponding combinatorics stays stable, would relate to vertex models on some random lattices/hypermaps. \\

\noindent \textbf{Organization of the paper:} The paper is organized as follows. In the first part, section \ref{sec:balanced-reminder}, we recall the results of \cite{DLN} in the balanced asymptotical regime. In a second part, section \ref{sec:unbalanced}, we introduce the necessary combinatorics to obtain our main result and then derive this results using a relevant combinatorial decomposition. The main theorem \ref{thm:main-thm} is a direct consequence of proposition \ref{prop:B-equations} and lemma \ref{lem:S-GF-equations} which are simple consequences of this combinatorial decomposition.
\section*{Acknowledgments}
I would like to thank Adrian Tanas\u a for an invitation to the LaBRI in July 2018. I would also like to thank Jean-Fran\c{c}ois Marckert for a discussion that has led to this letter. This work was supported by the Australian Research Council grant DP170102028. 
\section{The balanced regime}\label{sec:balanced-reminder}
\subsection{Freeness}
In the balanced regime studied in \cite{DLN} one considers marginals of a random quantum state $X\in \mathcal{H}_A\otimes \mathcal{H}_B\otimes \mathcal{H}_C\otimes \mathcal{H}_D$ such that $\dim \mathcal{H}_A=N_A$, $\dim \mathcal{H}_B=\dim \mathcal{H}_C=N$ and $\dim \mathcal{H}_D=N_D$. The balanced regime then corresponds to the regime where $N\rightarrow \infty$ and $N_A, N_D$ are arbitrary functions of $N$ such that $N_D(N)/N_A(N)\sim_{\infty}c>0$. Notice in particular that it is not necessary for $N_D$ and $N_A$ to grow to infinity while $N\rightarrow \infty$. The two type of marginals are 
\begin{equation}
V_{AB}=[\un \otimes \un \otimes\Tr\otimes\Tr](X\otimes X^\vee), \ V_{AC}=[\un \otimes \Tr \otimes\un\otimes\Tr](X\otimes X^\vee),
\end{equation}
where $X$ is a Gaussian complex tensor; see \cite{DLN} for details. Relying on the results of \cite{DLN}, we have the following theorem 
\begin{theorem}[D., Lionni, Nechita]\label{thm:freeness-thm}
Let $X\in \mathcal{H}_A\otimes \mathcal{H}_B\otimes \mathcal{H}_C\otimes \mathcal{H}_D$ be a sequence of random Gaussian tensors, where $N_{A,D}$ are arbitrary functions of $N$ satisfying $N_D\sim_\infty cN_{A}$ as $N\rightarrow \infty$, for some $c\in (0,\infty)$. Then the normalized marginals $(N_A^{-1}N^{-1}V_{AB}, N_{A}^{-1}N^{-1}V_{AC})$ converge in distribution, as $N\rightarrow \infty$, to a pair of identically distributed and free elements $(x_{AB}, x_{AC})$, where $x_{AB}$ and $x_{AC}$ have a $\mathrm{MP}_c$ distribution.
\end{theorem}
\subsection{ The resolvent at large $N$}
In what follows, we are interested in the generating function of moments of the random matrix\footnote{where the matrix product law is induced by the identification of $\cH_B$ with $\cH_C$.} $P=V_{AB}^{1/2} V_{AC} V_{AB}^{1/2}$, that is the Stieljtes transform of its eigenvalue density. We call this generating function the resolvent of $P$. More precisely we are interested in its large $N$ limit. That is, we want to compute
\begin{equation}
    W(z)=\lim_{\substack{N\rightarrow \infty\\ N_D(N)\sim_{\infty}c N_A(N)}}\sum_{n\ge 0}z^{-n-1}\frac{\mathbb{E}(\Tr(P^n))}{(N_A N)^{2n+1}}
\end{equation}\\
In this regime the resolvent can be computed in many ways. In particular, theorem \ref{thm:freeness-thm}, allows to compute it \textit{via} free probabilistic means as it implies that the large $N$ normalized\footnote{that is $\int \textrm{d}\rho_{\infty}=1$.} density of eigenvalues $\mathrm{d}\rho_{\infty}(x)$ of the matrix $P$ is the free multiplicative square of a Mar\u cenko-Pastur law of parameter $c$, that is $\mathrm{d}\rho_{\infty}(x)=\mathrm{d}\mathrm{MP}_c^{\boxtimes 2}(x)$.
The large $N$ resolvent can be computed as
\begin{equation}
    W(z)=\int_{-\infty}^{\infty}\frac{\mathrm{d}\rho_{\infty}(u)}{z-u}.
\end{equation}
\\

$W(z)$ can also be computed \textit{via} combinatorial arguments. Since the moments of the random matrix $P$ have a $1/N$ expansion (see \cite{DLN}), the corresponding resolvent also admits a $1/N$ expansion. The large $N$ limit of the moments is given by a specialization of the multi-variate Narayana statistics. This statement can be understood combinatorially thanks to a bijection of the graphs that describe the large $N$ limit moments, we refer to \cite{DLN} for a (very) short presentation of the bijection; see also \cite{lenczewski2013multivariate, DR-NarayanaWishart} for references relating the multi-variate Narayana statistics to the free probabilistic and random matrix context. In order to keep this letter succinct we refrain from presenting the bijective combinatorics proof. \\
Finally one can obtain the result of the proposition \ref{prop:balanced-resolvent} using an enumerative combinatorics technique similar to the one presented in the next section on the unbalanced case. In this balanced case the result is obtained more simply. We have
\begin{proposition}\label{prop:balanced-resolvent}
The large $N$ limit of the resolvent $W(z)$ satisfies the following algebraic equation
\begin{equation}
   z^2 W(z)^3 +2(c-1)zW(z)^2+\left((c-1)^2-z\right)W(z)+1=0.
\end{equation}
\end{proposition}
\begin{proof}
A proof is easily obtained \textit{via} free harmonic analysis. We give some details for readers who are not familiar with this technique here. From theorem \ref{thm:freeness-thm}, the normalized marginals converge in distribution to identically distributed free elements $x_{AB}, x_{AC}$. We use this fact to compute the algebraic equation satisfied by $W(z)$. Let us first define a few quantities. The moment generating functions 
\begin{equation}
    \chi_B(u)=\sum_{p\ge 1} u^p\varphi(x_{AB}^p), \quad
    \chi_C(u)=\sum_{p\ge 1} u^p\varphi(x_{AC}^p), \quad 
    \chi_{BC}(u)=\sum_{p\ge 1}u^p\varphi((x_{AC}x_{AB})^p),
\end{equation}
where $\varphi$ should be seen as the limit of the expectation of trace map $\bE \circ \Tr$. We also define the S-transforms
\begin{equation}
    S_{B}(t)=\frac{1+t}{t}\chi_B^{-1}(t), \quad S_{C}(t)=\frac{1+t}{t}\chi_C^{-1}(t), \quad S_{BC}(t)=\frac{1+t}{t}\chi_{BC}^{-1}(t),
\end{equation}
where the $\chi^{-1}(t)$ denote the reciprocal functions of the $\chi(t)$. Note that we have the relation
\begin{equation}\label{eq:W-chi-rel}
    W(z)=\frac1{z}+\frac{\chi_{BC}(1/z)}{z},
\end{equation}
due to the fact that $\lim_{N\rightarrow \infty}\frac{\mathbb{E}(\Tr(P^p))}{(N_A N)^{2p+1}}=\varphi((x_{AC}x_{AB})^p)$. Since $x_{AB}, x_{AC}$ are identically distributed,  $\chi_B(u)=\chi_C(u)$ and $S_{B}(t)=S_{C}(t)$, therefore we will denote $\chi(u)=\chi_B(u)=\chi_C(u)$. Moreover, we know\footnote{For instance, it can easily be obtained from \cite[Proof of Theorem 2.4]{DLN} from the relation between $G=W_{\textrm{MP}}$ and $\mathcal{R}$.} that the large $N$ limit of the resolvent\footnote{$\textrm{MP}$ stands for Mar\u cenko-Pastur.} $W_{\textrm{MP}}(z)$ of a Wishart matrix of rectangular parameter $c$ satisfies the algebraic equation
\begin{equation}
    z W_{\textrm{MP}}^2(z)+(c-z-1)W_{\textrm{MP}}(z)+1=0,
\end{equation}
this translates into an algebraic equation on $\chi(u)$ using the same relation \eqref{eq:W-chi-rel} for $W_{\textrm{MP}}(z)$ and $\chi(1/z)$
\begin{equation}
    u\chi(u)^2+((c+1)u-1)\chi(u)+ c u=0,
\end{equation}
which leads to an identity on $\chi^{-1}(t)$ by setting $u=\chi^{-1}(t)$ in the above equation, that is
\begin{equation}
    \chi^{-1}(t)=\frac{t}{t^2+(c+1)t+c}.
\end{equation}
Therefore, $S_{B}(t)=S_{C}(t)=\frac1{c+t}$. Since we have (see \cite[Chapter 3]{FreeRand}) $S_{BC}(t)=S_{B}(t)S_{C}(t)$ due to freeness of $x_{AB}, x_{AC}$ and $S_{BC}(t)=\frac{1+t}{t}\chi_{BC}^{-1}(t)$, we obtain an expression for $\chi_{BC}^{-1}(t)$
\begin{equation}
    \chi_{BC}^{-1}(t)=\frac{t}{(t+1) (c+t)^2},
\end{equation}
then leading to the following algebraic equation on $\chi_{BC}(u)$
\begin{equation}
    u (\chi_{BC}(u) +1) (c+\chi_{BC}(u) )^2-\chi_{BC}(u)=0.
\end{equation}
Finally we finish the proof by using the relation \eqref{eq:W-chi-rel} and setting $u=1/z$.
\end{proof}

\section{The unbalanced regime}\label{sec:unbalanced}
\subsection{Resolvent at large \texorpdfstring{$N_A$}{NA} and fixed \texorpdfstring{$N_B=N_C=m$}{NB=NC=m}}
In this section we are interested in the regime where $N_A$ is sent to infinity, while $N_D(N_A)\sim_{\infty} c N_A$ and $N_B=N_C=m$ are fixed to a constant value $m$. In this limit the moments can be written as a sum over combinatorial maps which are planar and whose detailed properties are explained later in this section.  In this section, we define the resolvent as being the following generating function of moments
\begin{equation}
W(z)=\lim_{\substack{N_A\rightarrow \infty\\N_D(N_A)\sim_{\infty} c N_A}}\sum_{n\ge0} z^{-n-1}\frac{\mathbb{E}(\Tr(P^n))}{(m N_A)^{2n+1}}.
\end{equation}
We show the following theorem.
\begin{theorem}\label{thm:main-thm}
The resolvent $W(z)$ satisfies the following algebraic equation,
\begin{multline}\label{eq:W-equation}
    W(z)^6 \left(y^4 z^4-2 y^2 z^4+z^4\right)+W(z)^5 \left(3 c y^4 z^3-6 c y^2 z^3+3 c z^3-3 y^4 z^3+6 y^2 z^3-3 z^3\right)\\
    +W(z)^4 \bigl(3 c^2 y^4 z^2-6 c^2 y^2 z^2+3 c^2 z^2-6 c y^4
   z^2+12 c y^2 z^2-6 c z^2+3 y^4 z^2-2 y^2 z^3-6 y^2 z^2\\
   -2 z^3+3 z^2\bigr)+W(z)^3 \bigl(c^3 y^4 z-2 c^3 y^2 z+c^3 z-3 c^2 y^4 z+6 c^2 y^2 z-3 c^2 z+3 c y^4 z-4 c y^2 z^2\\
   -6 c y^2 z-4 c
   z^2+3 c z-y^4 z+4 y^2 z^2+2 y^2 z+4 z^2-z\bigr)+W(z)^2 \bigl(-2 c^2 y^2 z-3 c^2 z+5 c y^2 z\\
   +5 c z -3 y^2 z+z^2-2 z\bigr)+W(z) \left(-c^3+c^2 y^2+2 c^2-2 c y^2+c
   z-c+y^2-z\right)-c=0.
\end{multline}
where $y=1/m$, and $m=\textrm{dim } \cH_B=\textrm{dim } \cH_C$ as introduced above and in \cite{DLN}.
\end{theorem}
 A nice feature of this equation is that one can find explicit roots in terms of radicals despite the fact that the equation is a polynomial of degree $6$ in $W(z)$. Moreover it also provides a recursive way to compute the moments. Indeed, it suffices to look for solutions analytic at infinity and to consider the induced equations on the coefficients of the expansion at infinity of these solutions. Let us illustrate this last fact by providing the first few moments, obtained using this method. Denoting $M_n=\frac{\mathbb{E}(\Tr(P^n))}{(m N_A)^{2n+1}}$, one has
 \begin{multline}
     \hspace{1.4cm} M_0=1, \ M_1=c^2+c y^2, \ M_2=c^4+4 c^3 y^2+2 c^3+2 c^2 y^4+4 c^2 y^2+c y^4,  \nonumber\\
     M_3= c^6+9 c^5 y^2+6 c^5+15 c^4 y^4+30 c^4 y^2+5 c^4+5 c^3 y^6+30 c^3 y^4+15 c^3 y^2+6 c^2 y^6+9 c^2 y^4+c y^6.
 \end{multline}\\
 
\noindent{\bf The combinatorial objects.}
The moments of order $k$ of the matrix $P$ can be expressed as a sum over bipartite labeled combinatorial maps with, one black vertex, up to $2k$ white vertices, $2k$ edges of two different types, type $0$ and type $1$ (denoted AB or AC, or B or C in \cite{DLN}). The two types of edge alternate around the black vertex. More formally, in terms of edge set and permutations (similar to \cite[Definition 2.9]{DLN}), these maps can be defined in the following way
\begin{definition}
The quadruple $\cM=(E,\sigma_{\bullet},\sigma_{\circ}, t)$ is called a \emph{combinatorial map with edge type} (combinatorial map or map for short). The set of edge $E$ is the set $E=\{1,2,3,\ldots, 2k\}$. The type map $t:E\rightarrow \{0,1\}$ sends edges of odd labels in $E_{o}=\{1,3,5,\ldots, 2k-1\}$ to $1$ giving them type $1$; while $t$ sends the edges of even labels $E=\{2,4,6,\ldots,2k\}$ to $0$, giving them type $0$. $\sigma_{\bullet}=(123\ldots 2k)$ is the full cycle permutation on $E$, while $\sigma_{\circ}$ is a permutation on $E$. The (unique) cycle of $\sigma_{\bullet}$ is the black vertex, while the cycles of $\sigma_{\circ}$ are the white vertices. The faces of $\cM$ are the cycles of $\sigma_{\bullet}\sigma_{\circ}$. We denote the set of such maps $\bM_k$; see Fig. \ref{fig:map-example-with-types} for an example. 
\end{definition}
\begin{figure}
    \centering
    \includegraphics{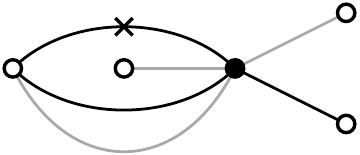}
    \caption{Example of a planar map $\cM$ with edge type. Edges of type $1$ are pictured in black while edges of type $0$ are pictured in dark gray. The edge of label $1$ is marked with a cross and edge labels increase counterclockwise around the black vertex. The corresponding permutation representation of the map is obtained by setting $\sigma_{\circ}=(2)(5)(6)(143)$ in complete cycle notation. The corresponding weight is $\textrm{wt}(\cM) = c^4 m^{-2}= c^4 y^2$.}
    \label{fig:map-example-with-types}
\end{figure}
Each map $\cM$ comes with a weight $\mathrm{wt}(\cM)$ and the moments of order $k$ of the matrix $P$ write as 
\begin{equation}
  \bE\left((N_AN)^{-2k-1}\Tr(P^k)\right) =  \sum_{\cM \in \bM_k}\mathrm{wt}(\cM).
\end{equation}
We do not give the general explicit weight function since we will only be interested in the weight of maps contributing to the large $N$ limit of the moments. Consequently we refer to \cite{DLN} for the general definition of the weight function.
In the large $N$ asymptotical regime, only planar maps contribute to the sum, that is the maps that satisfy the following constraint 
\begin{equation}\label{eq:planar-constraint}
    \#\sigma_{\circ}-2k+\#(\sigma_{\bullet}\sigma_{\circ})-1=0,
\end{equation}
where $\#\sigma_{\circ}$ is the number of cycles of $\sigma_{\circ}$ (\textit{i.e.} the number of white vertices of $\cM$), and $\#(\sigma_{\bullet}\sigma_{\circ})$ is the number of cycles of $\sigma_{\bullet}\sigma_{\circ}$ (\textit{i.e.} the number of faces of $\cM$). Note that we will denote $\bM_k^0$ the set of planar maps with $2k$ edges (that is satisfying the constraint of equation \eqref{eq:planar-constraint}). For these maps the weight function takes a simpler form thanks to \cite[Proposition 3.5 \& Theorem 3.16]{DLN}. We have,
\begin{equation}\label{eq:weight-fct}
    \mathrm{wt}(\cM)=\prod_{v_{\circ}\in \cM} c m^{-\mathrm{alt}(v_{\circ})},
\end{equation}
where $v_{\circ}$ runs among the white vertices of $\cM$, and $\mathrm{alt}(v_{\circ})$ was defined already in \cite{DLN}, and is the number of change of types of edge adjacent to $v_{\circ}$ when going around $v_{\circ}$. $m$ is a positive integer, and $c\in (0,\infty)$. As a consequence the resolvent $W(z)$ at large $N$ depends on both $m$ and $c$.\\

\noindent{\bf Petal decomposition.} We now introduce the combinatorial machinery needed to compute the generating function of the weighted planar combinatorial maps. First notice that the labeling of the edges of our combinatorial maps allows us to define a root. The root of our combinatorial maps is defined to be the edge $1 \in E$. Thus on Fig. \ref{fig:map-example-with-types}, the root edge is the edge marked with a cross. The petals $\cP$ of a planar map $\cM$ are defined thanks to the root edge as specific submaps of $\cM$. We define here the petals and the petal decomposition of a planar map. 
\begin{definition}
Let $\cM$ be a planar map in $\bM_k^0$. $\cM$ has one root edge. We define the petals of $\cM$ as follows. Let $e_{1}$ be the root edge of $\cM$. The complement edge $e_{1}^*$ of  $e_{1}$ is defined to be the edge next to $e_{1}$ when going around the white vertex counterclockwise adjacent to $e_{1}$. The petal $\cP_{1}$ associated to $e_{1}$ is the submap $\cP_{1}$ made of the black vertex, the edge $e_{1}^*$, and all the edges between $e_{1}$ and $e_{1}^*$ when turning around the black vertex counterclockwise as well as all the white vertices adjacent to these edges.  Consider the next edge $e_{2}$ counterclockwise around the black vertex after $e_{1}^*$. If $e_2\neq e_1$, we define its complement $e_{2}^*$ in the same way. The pair $(e_2,e_2^*)$ defines a petal $\cP_2$. Similarly we define petals $\cP_i$ for $i\in [\![ 1, q]\!]$ for some $q\le 2k$ such that $e_{q+1}=e_1$. The family of petals $\{\cP_i\}_{i=1}^q$ is the petal decomposition of $\cM$. Moreover, for each petal, we call external edges the pair of edge $(e_i,e_i^*)$.
\end{definition}
The existence and well-definedness of a petal decomposition for each planar maps is a simple consequence of the Jordan curve theorem that implies that the pair of external edge of a petal must separate two different regions of the plane.\\
Note that since petals are submaps, they are also maps. In particular, if we define the first external edge of a petal as its root, then a petal is a planar map\footnote{Since any submap of a planar map must be planar, \cite[Prop.~4.1.5, Sec.~4.1, p.~102]{GraphsOnSurfaces}.} such that the neighboring edge of the root edge counterclockwise around the white vertex adjacent to the root is also the neighboring edge of the root edge clockwise around the black vertex. We call $\bM^{0, \mathfrak P}$ the set of maps satisfying such constraint, while we denote  $\bM^{0, \mathfrak P}_n$ the set of such petals with $n$ edges. In particular, the weight of a petal is well defined and we can consider the generating series of weighted petals. This is what we do now.\\

\noindent{\bf Generatingfunctionology.} We use the petal decomposition to find the generating function of planar maps $\cM\in \bM^0$. To this aim we consider the generating functions of petals. There are four types of petals, depending on the type of their external edges. We denote $B_{ij}(x,y,c)$ the generating functions of the different types of petals for $i,j\in \{0,1\}$, where we set $y=1/m$, $x$ is the counting variable for the number of edges and $c$ counts the number of white vertices, we have
\begin{equation}
    B_{ij}(x,y,c)=\sum_{n\ge 1} x^n \sum_{\substack{\cP \in \bM^{0, \mathfrak P}_n \\ t(e)=i; t(\overline{e})=j}} \textrm{wt}(\cP),
\end{equation}
where $e$ and $\overline{e}$ denote the external edges of the petals.\\

Each planar map in $\bM^0$ can be seen as a finite sequence $(\cP_i)_{i=1}^q$ of petals, consequently the generating function of planar maps $A(x,y,c)=1+\sum_{n\ge 1} x^n \sum_{\substack{\cM \in \bM^0_n}} \textrm{wt}(\cM)$ can be written in terms of the generating functions of such sequences of petals. We denote these generating functions of sequences of petals $S_{ij}(x,y,c)$, where $i,j\in \{0,1\}$, and $j$ indicates the type of the first external edge of the first petal in the sequence, while $i$ indicates the type of the second external edge of the last petal in the sequence (sequences are read from right to left). Consequently the generating function of petals can be graphically represented as  \\
\begin{equation}
 S_{aa}(x,y,c)= \ \raisebox{-9.6mm}{\includegraphics[scale=0.8]{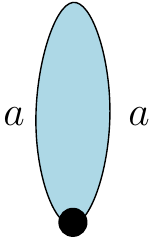}} \ = \sum_{n\ge 0} \sum_{\nu_i\in\{0,1\}_{i=1}^n} \ \raisebox{-9.6mm}{\includegraphics[scale=0.8]{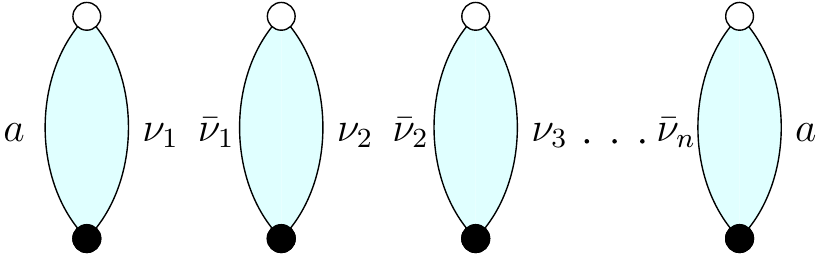}}
 \end{equation}
 \begin{equation}
 S_{\bar a a}(x,y,c)= \  \raisebox{-9.6mm}{\includegraphics[scale=0.8]{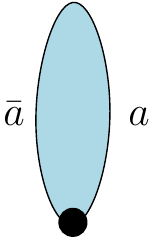}} \ =\raisebox{-0.5mm}{\includegraphics[scale=0.8]{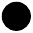}} \ +\sum_{n\ge 0}\sum_{\nu_i\in\{0,1\}_{i=1}^n} \ \raisebox{-9.6mm}{\includegraphics[scale=0.8]{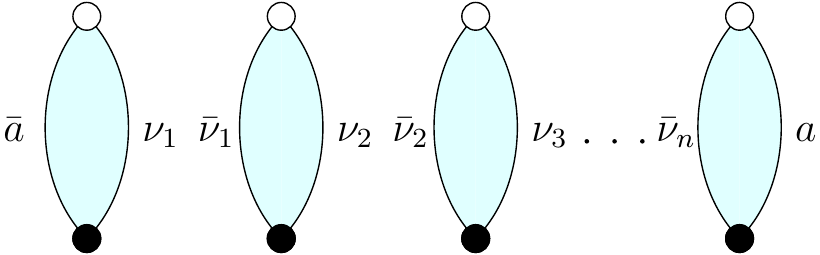}}\ ,
\end{equation}
where $a\in \{0,1\}$ and $\bar 0 = 1$ and $\bar 1 = 0$. More formally we have 
\begin{align}\label{}
  S_{a a}(x,y,c)&= B_{aa}(x,y,c)+\sum_{\nu_1\in \{0,1\} }B_{a\nu_1}(x,y,c)B_{\bar \nu_1 a}(x,y,c) +\ldots \nonumber\\
  \label{eq:S-GF-def1} &=\sum_{n\ge 0}\sum_{\nu_i\in\{0,1\}_{i=1}^n} B_{a\nu_1}(x,y,c)\left(\prod_{i=1}^{n-1} B_{\bar \nu_i\nu_{i+1}}(x,y,c)\right)B_{\bar \nu_n a}(x,y,c) \displaybreak[2]\\
  S_{\bar a a}(x,y,c)&=1+B_{\bar a a}(x,y,c)+\sum_{\nu_1\in \{0,1\} }B_{\bar a \nu_1}(x,y,c)B_{\bar \nu_1 a}(x,y,c) +\ldots \nonumber\displaybreak[2]\\
  \label{eq:S-GF-def2} &=1+\sum_{n\ge 0}\sum_{\nu_i\in\{0,1\}_{i=1}^n} B_{\bar a \nu_1}(x,y,c)\left(\prod_{i=1}^{n-1} B_{\bar \nu_i\nu_{i+1}}(x,y,c)\right)B_{\bar \nu_n a}(x,y,c),
\end{align}
where in the above sums the case $n=0$ corresponds to the one petal term, that is only one insertion of a $B$ generating function, while the term $n=1$ has only two factors, each being a $B$ generating function. In each of those cases the product in the summand is understood as trivial. \\
From the above arguments we have the following proposition
\begin{proposition}\label{prop:GFmoments-to-petal-sequences}
The generating function of planar maps $\cM\in \bM^0=\bigcup_{p\in \bN}\bM_p^0$ with the statistics induced by the weight function of \eqref{eq:weight-fct}, satisfies the following relation
\begin{equation}
    A(x,y,c)=S_{01}(x,y,c)=S_{10}(x,y,c).
\end{equation}
\end{proposition}
\begin{proof}
Indeed, using the decomposition in petals one finds that any planar map in $\bM^0$ can be decomposed as a sequence of petals. Given that the type of the root edge is fixed, the decomposition is unique and the last edge of the sequence of petals must have the opposite type ($0$). Thus $A(x,y,c)=S_{01}(x,y,c)$, and by symmetry under the change of type of edge in the generating functions, $A(x,y,c)=S_{10}(x,y,c)$.
\end{proof} 

We now need to write the relations between $B_{ij}(x,y,c)$ and $S_{ij}(x,y,c)$. We have the first set of equations
\begin{proposition}\label{prop:B-equations}
The $B_{ij}(x,y,c)$ satisfy the following relations with the $S_{ij}(x,y,c)$,
\begin{align}
 &B_{aa}(x,y,c) = cx+xS_{\bar{a}\bar{a}}(x,y,c)B_{aa}(x,y,c)+x S_{\bar{a}a}(x,y,c)B_{\bar{a}a}(x,y,c)\\
 &B_{a\bar a}(x,y,c)=xS_{\bar a \bar a}(x,y,c)B_{a \bar a}(x,y,c)+y^2x S_{\bar a a}(x,y,c)B_{\bar a \bar a}(x,y,c).
 \end{align}
\end{proposition}
\begin{proof}
We have the following graphical decomposition of the petals
 \begin{equation}\label{eq:graphical_rep_1}
 B_{aa}(x,y)=\ \raisebox{-7mm}{\includegraphics[scale=0.75]{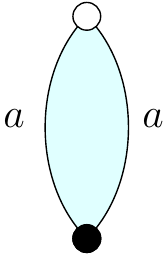}}\ =\raisebox{-7mm}{\includegraphics[scale=0.75]{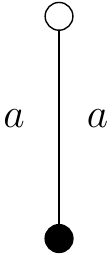}}\ +\ \raisebox{-7mm}{\includegraphics[scale=0.75]{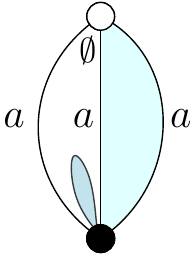}}\ +\ \raisebox{-7mm}{\includegraphics[scale=0.75]{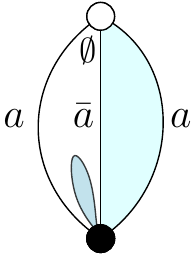}}
 \end{equation}
 \begin{equation}\label{eq:graphical_rep_2}
 B_{a\bar a}(x,y)=\ \raisebox{-7mm}{\includegraphics[scale=0.75]{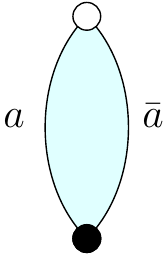}}\ = \ \raisebox{-7mm}{\includegraphics[scale=0.75]{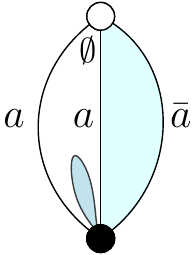}} \ + \ \raisebox{-7mm}{\includegraphics[scale=0.75]{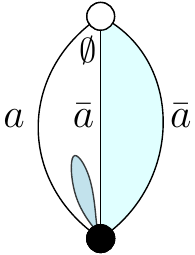}}
 \end{equation}
as indeed if one considers a petal whose external edges are of the same type, it either consists of one single edge, in which case the associated weight is $cx$, one $c$ for the unique white vertex and one $x$ for the unique edge, this leads to the first term of \eqref{eq:graphical_rep_1}. Or it consists of more than one edge in which case one can look at the second external edge $e^*$ and consider its next neighbouring edge $n(e^*)$ when going around the white vertex counterclockwise. If the type (denoted $a$) of $e^*$ is the same than the type of $n(e^*)$, $e^*$ and $n(e^*)$ enclose a non-trivial sequence of petals with starting edge type $\bar a$ and ending edge type $\bar a$, while $n(e^*)$ and $e$ induce a sub-petal whose external edge types are both $a$. This leads to the second term of \eqref{eq:graphical_rep_1}. The last term is obtained when the type of $n(e^*)$ is not the same than the type of $e^*$. In this case $e^*$ and $n(e^*)$ enclose a sequence of petals with starting edge type $a$ and ending edge type $\bar a$. This writes formally in terms of generating functions
\begin{equation}
  B_{aa}(x,y,c) = cx+xS_{\bar{a}\bar{a}}(x,y,c)B_{aa}(x,y,c)+x S_{\bar{a}a}(x,y,c)B_{\bar{a}a}(x,y,c).  
\end{equation}

Now consider a petal whose external edges are of different types. As previously, consider the neighboring edge $n(e^*)$ of the second external edge $e^*$. Either $n(e^*)$ is of the same type than $e^*$, say $a$. In this case the two edges enclose a sequence of petals with starting edge type $\bar a$ and ending edge type $\bar a$. This leads to the first term of \eqref{eq:graphical_rep_2}. The second term is obtained when $n(e^*)$ is of a different type than $e^*$. In this case the sequence of petals between the two edges has different starting and ending edge type, namely $a, \bar a$. Moreover, in this case the edge $e^*$, is located between the two edges of the sub-petal induced by $n(e^*)$ and $e$ and the edge type of $e^*$ is different than the edge type of $n(e^*)$ and $e$. Thus we have to take an additional factor $y^2$ into account. This leads to the following relation between generating functions
\begin{equation}
    B_{a\bar a}(x,y,c)=xS_{\bar a \bar a}(x,y,c)B_{a \bar a}(x,y,c)+y^2x S_{\bar a a}(x,y,c)B_{\bar a \bar a}(x,y,c).
\end{equation}
\end{proof}
The following lemma gives us an expression for $S_{aa}(x,y,c)$ and $S_{\bar a a}(x,y,c)$,
\begin{lemma}\label{lem:S-GF-equations}
After using the symmetries $B_{\bar a a}(x,y,c)=B_{a \bar a}(x,y,c)$, $B_{\bar a \bar a}(x,y,c)= B_{aa}(x,y,c)$ and similar symmetries for the $S$ generating functions, we have the following identities,
\begin{align}
    &S_{\bar a a}(x,y,c)=\frac{1-B_{01}(x,y,c)}{1-2B_{01}(x,y,c)+B_{01}(x,y,c)^2-B_{00}(x,y,c)^2}\\
    &S_{ a a}(x,y,c)=\frac{B_{00}(x,y,c)}{1-2B_{01}(x,y,c)+B_{01}(x,y,c)^2-B_{00}(x,y,c)^2}.
\end{align}
\end{lemma}
\begin{proof}
We prove this statement using a transfer matrix method. We define the $2\times 2$ matrix $T(x,y,c)$ element-wise such that $(T)_{ab}=B_{\bar a b}$ for $a,b \in\{0,1\}$ 
\begin{equation}
T(x,y,c)=\begin{pmatrix}B_{10}(x,y,c) & B_{11}(x,y,c) \\
B_{00}(x,y,c) & B_{01}(x,y,c)
\end{pmatrix}.
\end{equation}
From equations \eqref{eq:S-GF-def1}, \eqref{eq:S-GF-def2} we note that
\begin{align}
S_{aa}(x,y,c)&=(T(x,y,c))_{\bar a a}+(T(x,y,c)^2)_{\bar a a}+ \ldots\\
&= \sum_{n\ge 0} (T(x,y,c)^n)_{\bar a a}.
\end{align}
and
\begin{align}
S_{\bar a a}(x,y,c)&=1+(T(x,y,c))_{ a a}+(T(x,y,c)^2)_{ a a}+ \ldots \\
&=\sum_{n\ge 0} (T(x,y,c)^n)_{a a}.
\end{align}
These sums can be computed by diagonalizing $T(x,y,c)$. The eigenvalues of $T(x,y,c)$ write
\begin{multline}
\lambda_+=\frac12\Bigl(B_{10}(x,y,c)+B_{01}(x,y,c)\\
+\sqrt{(B_{10}(x,y,c)-B_{01}(x,y,c))^2+4B_{11}(x,y,c)B_{00}(x,y,c)}\Bigr)
\end{multline}
\begin{multline}
\lambda_-=\frac12\Bigl(B_{10}(x,y,c)+B_{01}(x,y,c)\\
-\sqrt{(B_{10}(x,y,c)-B_{01}(x,y,c))^2+4B_{11}(x,y,c)B_{00}(x,y,c)}\Bigr).
\end{multline}
Thanks to the symmetry $B_{\bar a a}(x,y,c)=B_{a \bar a}(x,y,c)$, $B_{\bar a \bar a}(x,y,c)= B_{aa}(x,y,c)$ the expressions of $\lambda_{\pm}(x,y)$ simplify to
\begin{align}
\lambda_{\pm}(x,y)=B_{01}(x,y)\pm B_{00}(x,y).
\end{align}
Moreover the symmetry also leads to a matrix of change of basis $Q$ independent of $x,y$ and $c$, that is we have 
\begin{equation}
Q=\begin{pmatrix}
-1 & 1 \\ 1& 1 
\end{pmatrix},
\end{equation}
thus we obtain
\begin{align}
\sum_{n\ge 0} T^n&= Q\begin{pmatrix} \frac1{1-B_{01}(x,y,c)+B_{00}(x,y,c)} & 0\\ 0 & \frac1{1-B_{01}(x,y,c)-B_{00}(x,y,c)}\end{pmatrix}Q^{-1}\\
&=\begin{pmatrix}\frac{1-B_{01}(x,y,c)}{(1-B_{01}(x,y,c))^2-B_{00}(x,y,c)^2} & \frac{B_{00}(x,y,c)}{(1-B_{01}(x,y,c))^2-B_{00}(x,y,c)^2} \\ \frac{B_{00}(x,y,c)}{(1-B_{01}(x,y,c))^2-B_{00}(x,y,c)^2} & \frac{1-B_{01}(x,y,c)}{(1-B_{01}(x,y,c))^2-B_{00}(x,y,c)^2}\end{pmatrix}
\end{align}
which leads to the result.
\end{proof}
\noindent{\bf Proof of Theorem \ref{thm:main-thm}.} As a consequence of these two results, proposition \ref{prop:B-equations} and lemma \ref{lem:S-GF-equations}, we have the following system of equations on $S_{01}$, $S_{00}$, $B_{01}$, $B_{00}$,
\begin{align}
    &S_{01}(x,y,c)(1-B_{01}(x,y,c)+B_{01}(x,y,c)^2-B_{00}(x,y,c)^2)+B_{01}(x,y,c)-1=0,\nonumber\\
    &S_{00}(x,y,c)(1-B_{01}(x,y,c)+B_{01}(x,y,c)^2-B_{00}(x,y,c)^2)-B_{00}(x,y,c)=0,\nonumber\\
    &cx+xS_{00}(x,y,c)B_{00}(x,y,c)+x S_{01}(x,y,c)B_{01}(x,y,c)-B_{00}(x,y,c)=0,\nonumber\\
    \label{eq:system_summary}&xS_{00}(x,y,c)B_{01}(x,y,c)+y^2x S_{01}(x,y,c)B_{00}(x,y,c)-B_{01}(x,y,c)=0.
\end{align}
Thanks to this system of polynomial equations we can obtain a polynomial equation on $S_{01}(x,y,c)$. Indeed, this system \eqref{eq:system_summary} defines an ideal $I$ in the ring of polynomials in seven variables $x,y,c,S_{00},S_{01},B_{00},B_{01}$, $$I\subseteq\mathbb{C}[x,y,c,S_{00},S_{01},B_{00},B_{01}].$$
We can obtain a polynomial equation for $S_{01}(x,y,c)$ by looking for a Gr\"obner basis for the elimination ideal $I_e=I\cap \mathbb{C}[x,y,c,S_{01}]$; see \cite{cox2013ideals} for Gr\"obner basis definition and properties. This Gr\"obner basis can be computed using a formal computation software. We used Magma \cite{Magma}. We find that $I_e=\left\langle \eta(S_{01},x,y,c)\right\rangle$ where $\eta$ consists of the following polynomial,
\begin{multline}
    \eta(S_{01},x,y,c)=S_{01}^7 \left(x^4 y^4-2 x^4 y^2+x^4\right)+S_{01}^6 \bigl(4 c x^4 y^4-8 c x^4 y^2+4 c x^4-3 x^4 y^4+6 x^4 y^2\\
    -3 x^4\bigr)
    +S_{01}^5 \bigl(6 c^2 x^4 y^4-12 c^2 x^4 y^2+6
   c^2 x^4-9 c x^4 y^4+18 c x^4 y^2-9 c x^4+3 x^4 y^4-6 x^4 y^2+3 x^4-2 x^2 y^2\\-2 x^2\bigr)
   +S_{01}^4 \bigl(4 c^3 x^4 y^4-8 c^3 x^4 y^2+4 c^3 x^4-9 c^2 x^4 y^4+18 c^2 x^4 y^2-9 c^2
   x^4+6 c x^4 y^4
   -12 c x^4 y^2+6 c x^4\\
   -6 c x^2 y^2-6 c x^2-x^4 y^4+2 x^4 y^2-x^4+4 x^2 y^2+4 x^2\bigr)
   +S_{01}^3 \bigl(c^4 x^4 y^4-2 c^4 x^4 y^2+c^4 x^4
   -3 c^3 x^4 y^4\\
   +6 c^3 x^4
   y^2-3 c^3 x^4+3 c^2 x^4 y^4-6 c^2 x^4 y^2+3 c^2 x^4-6 c^2 x^2 y^2-7 c^2 x^2-c x^4 y^4+2 c x^4 y^2-c x^4+9 c x^2 y^2+9 c x^2\\
   -3 x^2 y^2-2 x^2+1\bigr)
   +S_{01}^2 \left(-2 c^3 x^2
   y^2-4 c^3 x^2+6 c^2 x^2 y^2+7 c^2 x^2-5 c x^2 y^2-3 c x^2+2 c+x^2 y^2-1\right)\\
   +S_{01} \left( -c^4x^2+c^3 x^2 y^2+2 c^3 x^2-2 c^2 x^2 y^2-c^2 x^2+c^2+c x^2 y^2-2
   c\right)-c^2
\end{multline}
The solutions $S_{01}(x,y,c)$ of the equation $\eta(S_{01},x,y,c)=0$ consists of the projection of the solutions for fixed values of $x,y,c$ to the system \eqref{eq:system_summary} onto the $\bC$-plane in $\bC^7$ corresponding to the variable $S_{01}$. We refer to \cite[Chapter 2 \& Chapter 3]{cox2013ideals} for explanation of this method.\\
Note that the polynomial $\eta$ seen as a univariate polynomial of $S_{0,1}$ (when fixing the values of $x,y,c$) is divisible by the polynomial $(c+S_{0,1})$. Consequently $S_{0,1}=-c$ is a solution. However this is not an interesting solution for us as it would lead to a resolvent $W(z)=-c/z$, that is, all moments of order greater than zero vanish and the moment of order zero is negative. We can factor this uninteresting solution and find a degree $6$ polynomial on $S_{01}$. In order to find the polynomial equation of Theorem \ref{thm:main-thm} satisfied by the resolvent we then need to perform the change of variables $x\rightarrow1/\sqrt{z}$ and $S_{01}(1/\sqrt{z},y,c)\rightarrow z W(z)$. Indeed, we know from proposition \ref{prop:GFmoments-to-petal-sequences} that the generating function $A(x,y,c)$ of planar maps in $\bM^0$ is simply $S_{01}(x,y,c)$. Moreover,  the variable $1/z$ counts pairs of edges of alternating type while the variable $x$ counts edges. This explains the square root.\\
To avoid too cumbersome expression we also keep the dependence of $W(z)$ in $y,c$ implicit, thus showing Theorem \ref{thm:main-thm}. \qed \\

Despite the apparent complexity of the algebraic equation \eqref{eq:W-equation}, it is possible to give explicit solutions in terms of radicals. We do not present them here but they can be obtained \textit{via} symbolic computation softwares. However, since there are six solutions we need to select the solution which is the generating function of moments of our matrix $P$. This is done by requiring that the correct solution should be analytic at $z=\infty$ and requiring that the first coefficient of its expansion at infinity is $1$. \\
\begin{remark}
Note that, due to the result of Theorem \ref{thm:freeness-thm}, we expect that in the limit $y\rightarrow 0$, $c\rightarrow 1$, equation \eqref{eq:W-equation} reduces to 
\begin{equation}\label{eq:Fuss-Catalan-case}
    z^2W(z)^3-z W(z)+1=0.
\end{equation}
It turns out that this is not the case. In this limit, the equation \eqref{eq:W-equation} simplifies into the polynomial equation
\begin{equation}
   W(z)^6 z^4-2 W(z)^4 z^3+W(z)^2 z^2-1=0.
\end{equation}
which in turns factors into
\begin{equation}\label{eq:curve-splitting-y=0}
    \left(W(z)^3 z^2-W(z) z-1\right) \left(W(z)^3 z^2-W(z) z+1\right)=0.
\end{equation}
Only the second factor relates to the equation \eqref{eq:Fuss-Catalan-case}. This non-trivial behavior is due to the fact that equation \eqref{eq:W-equation} defines a curve of generic genus two. In the limit $y\rightarrow 0$, $c\rightarrow 1$, the curve degenerates into two non-trivial connected components of genus zero. Only one of these components is a Fuss-Catalan curve defined by the polynomial equation \eqref{eq:Fuss-Catalan-case}. Consequently we need to select the second factor of \eqref{eq:curve-splitting-y=0}. Note in particular that sole the second factor has a solution analytic at $z=\infty$.
\end{remark}
Note also that the position of the ramification points in $z$ of the solutions to the polynomial equation \eqref{eq:W-equation} can be computed exactly in terms of $y$ and $c$. This allows us to infer the support of the corresponding eigenvalue density. However their expressions being quite involved we do not display them here.  \\

In principle, we can obtain the exact expression of the density by using the inverse Stieltjes transform formula. However the direct computation of the limit seems intractable. One can also obtain a polynomial equation on the density from the polynomial equation \eqref{eq:W-equation} on the resolvent using a technique described for instance in \cite{dartois2019schwinger} and polarization formula. Unfortunately the polynomial equation on the density thus obtained is of very high order and it does not seem possible to obtain the relevant root\footnote{that is the root which leads to a positive normalized density} exactly in terms of radicals.

\bibliographystyle{alpha}
\bibliography{Typical-RMT-biblio}
\end{document}